\documentclass[conference,letterpaper,romanappendices]{ieeeconf}
\let\proof\relax   

\usepackage{amsthm,xpatch}
\usepackage{amsmath,amsfonts}
\usepackage{cite}
\usepackage{amssymb}
\usepackage{dsfont}
\usepackage{graphicx, subfigure}
\usepackage{color}
\usepackage{breqn}
\usepackage{mathtools}
\usepackage{bbm}
\usepackage{latexsym}
\usepackage[ruled, linesnumbered]{algorithm2e}
\usepackage{accents}
\usepackage{tikz}

\newtheorem{lemma}{Lemma}
\newtheorem{theorem}{Theorem}

\newtheorem{example}{Example}

\DeclareMathOperator*{\argmin}{\arg\!\min}

\newcommand*{\transpose}{%
  {\mathpalette\@transpose{}}%
}

\IEEEoverridecommandlockouts

\begin{document}

\newcommand{\SB}[3]{
\sum_{#2 \in #1}\biggl|\overline{X}_{#2}\biggr| #3
\biggl|\bigcap_{#2 \notin #1}\overline{X}_{#2}\biggr|
}

\newcommand{\Mod}[1]{\ (\textup{mod}\ #1)}

\newcommand{\overbar}[1]{\mkern 0mu\overline{\mkern-0mu#1\mkern-8.5mu}\mkern 6mu}

\makeatletter
\newcommand*\nss[3]{%
  \begingroup
  \setbox0\hbox{$\m@th\scriptstyle\cramped{#2}$}%
  \setbox2\hbox{$\m@th\scriptstyle#3$}%
  \dimen@=\fontdimen8\textfont3
  \multiply\dimen@ by 4             
  \advance \dimen@ by \ht0
  \advance \dimen@ by -\fontdimen17\textfont2
  \@tempdima=\fontdimen5\textfont2  
  \multiply\@tempdima by 4
  \divide  \@tempdima by 5          
  \ifdim\dimen@<\@tempdima
    \ht0=0pt                        
    \@tempdima=\fontdimen5\textfont2
    \divide\@tempdima by 4          
    \advance \dimen@ by -\@tempdima 
    \ifdim\dimen@>0pt
      \@tempdima=\dp2
      \advance\@tempdima by \dimen@
      \dp2=\@tempdima
    \fi
  \fi
  #1_{\box0}^{\box2}%
  \endgroup
  }
\makeatother

\makeatletter
\renewenvironment{proof}[1][\proofname]{\par
  \pushQED{\qed}%
  \normalfont \topsep6\p@\@plus6\p@\relax
  \trivlist
  \item[\hskip\labelsep
        \itshape
    #1\@addpunct{:}]\ignorespaces
}{%
  \popQED\endtrivlist\@endpefalse
}
\makeatother

\makeatletter
\newsavebox\myboxA
\newsavebox\myboxB
\newlength\mylenA

\newcommand*\xoverline[2][0.75]{%
    \sbox{\myboxA}{$\m@th#2$}%
    \setbox\myboxB\null
    \ht\myboxB=\ht\myboxA%
    \dp\myboxB=\dp\myboxA%
    \wd\myboxB=#1\wd\myboxA
    \sbox\myboxB{$\m@th\overline{\copy\myboxB}$}
    \setlength\mylenA{\the\wd\myboxA}
    \addtolength\mylenA{-\the\wd\myboxB}%
    \ifdim\wd\myboxB<\wd\myboxA%
       \rlap{\hskip 0.5\mylenA\usebox\myboxB}{\usebox\myboxA}%
    \else
        \hskip -0.5\mylenA\rlap{\usebox\myboxA}{\hskip 0.5\mylenA\usebox\myboxB}%
    \fi}
\makeatother

\xpatchcmd{\proof}{\hskip\labelsep}{\hskip3.75\labelsep}{}{}

\pagestyle{plain}

\title{\fontsize{22.59}{28}\selectfont A Simple and Efficient Strategy for the Coin Weighing Problem with a Spring Scale}

\author{Esmaeil Karimi, Fatemeh Kazemi, Anoosheh Heidarzadeh, and Alex Sprintson\thanks{The authors are with the Department of Electrical and Computer Engineering, Texas A\&M University, College Station, TX 77843 USA (E-mail: \{esmaeil.karimi, fatemeh.kazemi, anoosheh, spalex\}@tamu.edu).}}


\maketitle 

\thispagestyle{plain}

\begin{abstract}
This paper considers a generalized version of the coin weighing problem with a spring scale that lies at the intersection of group testing and compressed sensing problems. Given a collection of $n\geq 2$ coins of total weight $d$ (for a known integer $d$), where the weight of each coin is an unknown integer in the range of $\{0,1,\dots,k\}$ (for a known integer $k\geq 1$), the problem is to determine the weight of each coin by weighing subsets of coins in a spring scale. The goal is to minimize the average number of weighings over all possible weight configurations. For $d=k=1$, an adaptive bisecting weighing strategy is known to be optimal. However, even the case of $d=k=2$, which is the simplest non-trivial case of the problem, is still open. For this case, we propose and analyze a simple and effective adaptive weighing strategy. A numerical evaluation of the exact recursive formulas, derived for the analysis of the proposed strategy, shows that this strategy requires about ${1.365\log_2 n -0.5}$ weighings on average. To the best of our knowledge, this is the first non-trivial achievable upper bound on the minimum expected required number of weighings for the case of $d=k=2$. As $n$ grows unbounded, the proposed strategy, when compared to an optimal strategy within the commonly-used class of nested strategies, requires about $31.75\%$ less number of weighings on average; and in comparison with the information-theoretic lower bound, it requires at most about $8.16\%$ extra number of weighings on average. 
\end{abstract}

\section{introduction}
In this work, we consider a generalized version of the coin weighing (CW) problem with a spring scale~\cite{B:09}. Suppose that there is a collection of $n\geq 2$ coins of total weight $d$, where each coin has an unknown integer weight in the set $\{0,1\dots,k\}$, for some known integers $d\geq 1$ and ${k\geq 1}$. The goal is to determine the weight of each coin by weighing subsets of coins in a spring scale. The problem is to devise an adaptive weighing strategy, where each weighing can depend on the results of the previous weighings, that minimizes (i) the maximum number of required weighings over all possible weight configurations (worst-case setting), or (ii) the average number of required weighings over all possible weight configurations (average-case setting). 




The CW problem lies at the intersection of group testing and compressed sensing problems. In particular, for $k=1$ and $d\leq n$, the CW problem is equivalent to the combinatorial quantitative group testing problem, see, e.g.,~\cite{DH:2000}. 
Also, for $d\ll n$ and $k\geq 1$, the CW problem is equivalent to the integral compressed sensing problem where both the signal and the sensing matrix are integer valued, see, e.g.,~\cite{KKLP:17}. 

For $d=k=1$, a simple adaptive bisecting weighing strategy is optimal in both worst-case and average-case settings~\cite{WZC:17}. However, the simplest non-trivial case of the problem, i.e., $d=k=2$, is still open, and hence the focus of this work. For the worst-case setting, a simple information-theoretic argument yields a lower bound on the minimum required number of weighings by $\max\{\log_2 n,\log_3\binom{n}{2}\}$ (see Theorem~\ref{lem:bound}); and for the average-case setting, a similar argument gives a lower bound of $\frac{2}{n+1}\log_2 n+\frac{n-1}{n+1}\log_3\binom{n}{2}$ on the minimum expected required number of weighings (see Theorem~\ref{lem:bound}). Notwithstanding, the question whether these lower bounds are achievable remains open. For the worst-case setting, ${2\log_2 n-1}$ weighings are known to be sufficient, and this bound is achievable by a simple nested strategy (see~\cite[Lemma~1]{WZC:17}). This quantity also serves as an upper bound for the average-case setting, and no tighter achievable upper bound was previously reported.

\subsection{Related Work and Applications}
The worst-case setting of the CW problem was originally proposed in~\cite{S:60} for $k=1$ and unknown $d$, and was later studied for $k=1$ and known $d$, e.g., in~\cite{L:75,AS:85,A:86}. Various order-optimal strategies were previously proposed for unknown $d$, see, e.g.,~\cite{L:75,MK:90}, and for known $d$, see, e.g.,~\cite{A:88,UTW:2000,B:09,EM:14,ARP:17}. Recently, in~\cite{B:09}, Bshouty proposed the first and only known order-optimal strategy for any $k>1$ and unknown $d$, and no such result exists for any $k>1$ and known $d$. Despite the rich literature on the worst-case setting, there was no result for the average-case setting of the CW problem prior to the present work, excluding the results that trivially carry over from worst case into average case.  

The worst-case setting of the CW problem has also been extensively studied for a wide range of applications, e.g., multi-access communication, spectrum sensing, traffic monitoring, anomaly detection, and network tomography, to name a few (see, e.g.,~\cite{WZC:17}, and references therein). Moreover, most of these applications are being run repeatedly over time, and for such applications, the average-case performance is expected to be more relevant than the worst-case performance. This observation is the primary motivation for studying the average-case setting of the CW problem in this work. 


\subsection{Main Contributions}
In this work, we propose and analyze a simple and effective adaptive weighing strategy for ${d=k=2}$. The results of our theoretical analysis show that the proposed strategy requires ${2\log_2 n-1}$ number of weighings in worst case, and it requires about ${1.365\log n -0.5}$ number of weighings on average. (The average-case result is obtained by a numerical evaluation of the exact recursive formulas, derived for the analysis of performance of the proposed strategy.) This is the first non-trivial achievable upper bound on the minimum expected required number of weighings for $d=k=2$. Additionally, for the average-case setting, we design and analyze an optimal strategy within the class of nested strategies, which are mostly being used in today's applications, that requires ${\frac{2n+1}{n+1}\log n-\frac{2(n-1)}{n+1}}$ weighings on average. A simple analysis shows that as $n$ grows unbounded, the proposed strategy, when compared to the optimal nested strategy, requires about $31.75\%$ less number of weighings on average; and when compared to the information-theoretic lower bound, the proposed strategy requires at most about $8.16\%$ extra number of weighings on average.





%


\section{Setup and Notations}\label{sec:SN}
Fix an integer $l\geq 1$, and let $n=2^{l}$. Let $N = \{1,\dots,n\}$. Consider a collection $N$ of $n$ coins, each coin $i\in N$ of an unknown integer weight $w_i\in \{0,1,2\}$. We refer to the set $\{w_1,\dots,w_n\}$, simply denoted by $\{w_i\}$, as the \emph{weight configuration}, or the \emph{configuration}, for short. For any $S\subseteq N$, denote by $w(S)$ the total weight of the subset $S$ of coins, i.e., $w(S) = \sum_{i\in S} w_i$. We assume that the total weight of $N$, i.e., $w(N)$, is equal to $2$. 

The problem is to determine the weight of all coins in $N$ by weighing subsets of $N$ in a spring scale. In the worst-case setting of the problem, the goal is to minimize the maximum number of required weighings over all possible configurations; and in the average-case setting of the problem, the goal is to minimize the expected number of required weighings over all possible configurations, where all possible configurations are assumed to be equally probable.  

Since $w(N)=2$ and $w_i\in\{0,1,2\}$ for all $i\in N$, there are $n$ distinct configurations such that $w_i = 2$ for some $i\in N$, and $w_j = 0$ for all $j\in N\setminus \{i\}$, and there are $\binom{n}{2}$ distinct configurations such that $w_{i}=w_{j} = 1$ for some $i,j\in N$ and $w_k = 0$ for all ${k\in N\setminus \{i,j\}}$. We refer to the first group of configurations as {\emph{Type-I}}, and refer to the second group as {\emph{Type-II}}. For example, for $n=2$, the possible configurations $\{w_1,w_2\}$ are $\{2,0\}$, $\{0,2\}$, and $\{1,1\}$, where the first two configurations are Type-I and the third one is Type-II. For the ease of exposition, we define a representative function $\Delta(\{w_i\}_{i\in S})$ for any $S\subseteq N$, $w(S)=2$, as follows. For any Type-I (sub-) configuration $\{w_i\}_{i\in S}$, $\Delta(\{w_i\}_{i\in S})=0$, and for any Type-II (sub-) configuration $\{w_i\}_{i\in S}$, ${\Delta(\{w_i\}_{i\in S}) =|i-j|}$, where $w_i=w_j=1$. 

Any adaptive weighing strategy $\Psi$ can be defined as a sequence $\{S_1,S_2,\dots\}$ of subsets of coins that are to be weighed following the prescribed order, where the choice of each subset $S_i$ can depend on $\{S_j\}_{j=1}^{i-1}$ and $\{w(S_j)\}_{j=1}^{i-1}$. 

Consider an arbitrary strategy $\Psi$. Denote by $T^{\Psi}_{\mathrm{ave}}(n)$ the expected number of weighings required by the strategy $\Psi$ to determine the weight of all coins in $N$, over all possible weight configurations. For any subset $S$ of coins, all with unknown weights, we denote by $T_w^{\Psi}(s)$ the expected number of weighings that the strategy $\Psi$ performs to determine the weight of all coins in $S$, where $s = |S|$ and $w = w(S)$. The expectation is taken over all possible (sub-) configurations $\{\tilde{w}_{i}\}_{i\in S}$, $\tilde{w}_i\in \{0,1,2\}$, such that ${\sum_{i\in S} \tilde{w}_i=w}$.


For any subset $S$ of coins, all with unknown weights, such that $w(S)=2$, denote by $T^{\Psi}(s|\Delta)$ the expected number of weighings that the strategy $\Psi$ performs to determine the weight of all coins in $S$, given that $\Delta(\{w_i\}_{i\in S}) = \Delta$, where $s = |S|$. Here, the expectation is taken over all possible (sub-) configurations $\{\tilde{w}_{i}\}_{i\in S}$, $\tilde{w}_i\in \{0,1,2\}$, such that ${\sum_{i\in S} \tilde{w}_i=2}$ and $\Delta(\{\tilde{w}_i\}_{i\in S}) = \Delta$. 

For any disjoint subsets $A$ and $B$ of coins, all with unknown weights, such that $w(A)=1$ and $w(B)=1$, denote by $T^{\Psi}(a,b)$ the expected number of weighings required by the strategy $\Psi$ to determine the weight of all coins in $A$ and $B$, where $a=|A|$ and $b=|B|$. The expectation is here taken over all possible (sub-) configurations $\{\tilde{w}_{i}\}_{i\in A}$ and $\{\tilde{w}_{i}\}_{i\in B}$, $\tilde{w}_i\in \{0,1\}$, such that ${\sum_{i\in A} \tilde{w}_i=1}$ and ${\sum_{i\in B} \tilde{w}_i=1}$. For convenience, we adopt the convention $T^{\Psi}(1,s) = T^{\Psi}(s,1)= T^{\Psi}_1 (s)$. 

From now on, whenever the strategy $\Psi$ is clear from the context, we omit the superscript $\Psi$, and denote $T^{\Psi}_{\mathrm{ave}}(n)$, $T_{w}^{\Psi}(s)$, $T^{\Psi}(s|\Delta)$, and $T^{\Psi}(a,b)$ by $T_{\mathrm{ave}}(n)$, $T_{w}(s)$, $T(s|\Delta)$, and $T(a,b)$, respectively. Moreover, we define $T_{\mathrm{max}}(n)$, $T_{w}^{\star}(s)$, $T^{\star}(s|\Delta)$, and $T^{\star}(a,b)$ similarly as $T_{\mathrm{ave}}(n)$, $T_{w}(s)$, $T(s|\Delta)$, and $T(a,b)$, respectively, except for the maximum number of weighings, instead of the expected number of weighings, that the strategy $\Psi$ must perform. 

\begin{theorem}\label{lem:bound}
For any weighing strategy $\Psi$, we have \[T^{\Psi}_{\mathrm{max}}(n)\geq \max\left\{\log_2 n,\log_3\binom{n}{2}\right\}\] and \[T^{\Psi}_{\mathrm{ave}}(n)\geq \frac{2}{n+1}\log_2 n+\frac{n-1}{n+1}\log_3\binom{n}{2}.\]
\end{theorem}

\begin{proof}
Recall that there are two types of weight configurations: Type-I and Type-II. For any Type-I configuration, the result of weighing on any subset of coins is either zero or non-zero, and the number of distinct possible configurations of Type-I is $n$. Thus, at least $\log_2 n$ weighings are needed to distinguish a particular configuration of this type. For any Type-II configuration, the result of weighing on any subset of coins can be $0$, or $1$, or $2$. Thus, there are $\binom{n}{2}$ distinct possible configurations of this type, and to distinguish a particular configuration within this class, one needs at least $\log_3 \binom{n}{2}$ weighings. Accordingly, for a configuration of an unknown type, at least $\max\{\log_2 n,\log_3\binom{n}{2}\}$ weighings are required to identify the configuration. Since all configurations are equally probable, it can be easily verified that a randomly chosen configuration is of Type-I or of Type-II with probability $\frac{2}{n+1}$ or $\frac{n-1}{n+1}$, respectively. Consequently, on average, at least $\frac{2}{n+1}\log_2 n+\frac{n-1}{n+1}\log_3\binom{n}{2}$ weighings are necessary to identify a particular configuration of an unknown type. 
\end{proof}



\section{Proposed Weighing Strategy}\label{sec:PS}
In this section, we propose a weighing strategy that determines the weight of all coins, for the setup in Section~\ref{sec:SN}. 

For any set $S=\{i_1,\dots,i_{|S|}\}$ such that $|S|$ is a power of $2$, we denote by $S_1$ and $S_2$ the two disjoint subsets $\{i_1,\dots,i_{|S|/2}\}$ and $\{i_{|S|/2+1},\dots,i_{|S|}\}$, respectively. 

The proposed strategy is based on three recursive procedures $\Pi_0$, $\Pi_1$, and $\Pi_2$, described shortly. At the beginning, the strategy initializes the weight of all coins by zero, i.e., $\hat{w}_i=0$ for all $i\in N$. Then, it starts with the procedure $\Pi_0$ over the set $N$. The weights of coins will be updated recursively according to the procedures $\Pi_0$, $\Pi_1$, and $\Pi_2$. This process is terminated once the sum of weights of all coins, $\sum_{i\in N} \hat{w}_i$, is equal to $2$, and the strategy returns $\{\hat{w}_i\}_{i\in N}$.  

The inputs of the procedure $\Pi_0$ are a set $S$ and its weight $w(S)$. The procedure $\Pi_1$ takes as input two disjoint sets $A$ and $B$ such that $w(A)=w(B)=1$, and the procedure $\Pi_2$ takes as input two disjoint sets $A$ and $B$ such that $w(A)=w(B)=w(A_1\cup B_1)=1$. (Recall that $A_1=\{i_1,\dots,i_{|A|/2}\}$ and $B_1=\{j_1,\dots,j_{|B|/2}\}$ when $A = \{i_1,\dots,i_{|A|}\}$ and $B = \{j_1,\dots,j_{|B|}\}$.) We represent these procedures by $\Pi_0(S)$, $\Pi_1(A,B)$, and $\Pi_2(A,B)$, respectively. 

\subsection{Procedure $\Pi_0$}
For any $S=\{i\}$, the procedure $\Pi_0(S)$ updates $\hat{w}_i$ by $w(S)$; and for any $S$, $|S|>1$, the procedure $\Pi_0(S)$ begins with weighing $S_1$. If $w(S_1)=0$ or ${w(S_1)=2}$, the procedure $\Pi_0(S)$ continues with $\Pi_0(S_2)$ or $\Pi_0(S_1)$, respectively. Otherwise, depending on $w(S)=1$ or $w(S)=2$, the procedure $\Pi_0(S)$ continues with $\Pi_0(S_1)$ or $\Pi_1(S_1,S_2)$, respectively. We note that for $w(S_1)=0$ or $w(S_1)=2$, the procedure $\Pi_0$ follows a simple bisecting strategy, and for $w(S_1)=1$, it follows a generalized bisecting strategy defined below.

\subsection{Procedure $\Pi_1$}
For any $A=\{i\}$ and $B=\{j\}$, the procedure $\Pi_1(A,B)$ updates $\hat{w}_i$ and $\hat{w}_j$ by $1$; For any $A$ and $B$ such that $|A|=1$ and $|B|>1$ or $|A|>1$ and $|B|=1$, the procedure $\Pi_1(A,B)$ continues with two procedures $\Pi_0(A)$ and $\Pi_0(B)$. For any $A$ and $B$ such that $|A|>1$ and $|B|>1$, the procedure $\Pi_1(A,B)$ weighs $A_1\cup B_1$. If $w(A_1\cup B_1)=0$ or $w(A_1\cup B_1)=2$, the procedure $\Pi_1(A,B)$ continues with $\Pi_1(A_2,B_2)$ or $\Pi_1(A_1,B_1)$, respectively; otherwise, it continues with $\Pi_2(A,B)$. 

\subsection{Procedure $\Pi_2$}
For any $A=\{i_1,i_2\}$ and $B=\{j_1,j_2\}$, the procedure $\Pi_2(A,B)$ weighs $A_1=\{i_1\}$, and updates $\hat{w}_{i_1}$, $\hat{w}_{i_2}$, $\hat{w}_{j_1}$, and $\hat{w}_{j_2}$ by $w(A_1)$, $1-w(A_1)$, $1-w(A_1)$, and $w(A_1)$, respectively. For any $A$ and $B$ such that $\max(|A|,|B|)>2$ and $|A|\leq |B|$, the procedure $\Pi_2(A,B)$ weighs $A_1\cup (B_2)_1$. (Recall that $(B_2)_1 = \{j_{|B|/2+1},\dots,j_{3|B|/4}\}$ when $B_2 = \{j_{|B|/2+1},\dots,j_{|B|}\}$.) If $w(A_1\cup (B_2)_1)$ is equal to $0$, $1$, or $2$, the procedure $\Pi_2(A,B)$ continues with $\Pi_1(A_2,B_1)$, $\Pi_1(A_1,(B_2)_2)$, or $\Pi_1(A_1,(B_2)_1)$, respectively. For any $A$ and $B$ such that $\max(|A|,|B|)>2$ and $|B|<|A|$, the procedure is the same, except for $A$ and $B$ being interchanged.   

\begin{example}
Consider $n=8$ coins of weights $w_3=w_6=1$ and $w_i=0$ for all $i\not\in \{3,6\}$. Let $N = \{1,\dots,8\}$. 	Initialize $\hat{w}_i$ by $0$ for all $i\in N$. Applying $\Pi_0(N)$, the set $\{1,\dots,4\}$ is weighed. Since $w(\{1,\dots,4\})=1$, the strategy proceeds with $\Pi_1(\{1,2,3,4\},\{5,6,7,8\})$. According to the strategy, the set $\{1,2\}\cup \{5,6\}$ is weighed. Since ${w(\{1,2\}\cup \{5,6\})=1}$ the strategy continues with $\Pi_2(\{1,2,3,4\},\{5,6,7,8\})$. According to the procedure $\Pi_2$, weighing is performed on $\{1,2\}\cup \{7\}$. Since $w(\{1,2\}\cup \{7\})=0$, the strategy proceeds with $\Pi_1(\{3,4\},\{5,6\})$, and weighs $\{3\}\cup \{5\}$. Since $w(\{3\}\cup \{5\})=1$, the strategy continues with $\Pi_2(\{3,4\},\{5,6\})$. According to the procedure $\Pi_2$, the weighing is performed on $\{3\}$. Since $w(\{3\})=1$, the strategy updates $\hat{w}_3=1$, $\hat{w}_4=0$, $\hat{w}_5=0$, and $\hat{w}_6=1$. Since $\sum_{i\in N} \hat{w}_i=2$, the process is terminated.       
\end{example}

\section{Analysis of the Proposed Startegy}\label{sec:Analysis}
In this section, we analyze the average-case and worst-case performance of the strategy proposed in Section~\ref{sec:PS}. 

For simplifying the notation, for all $0\leq i,j\leq l$, we denote $T(2^{i},2^{j})$ and $T^{\star}(2^{i},2^{j})$ by $T_{i,j}$ and $T^{\star}_{i,j}$, respectively. 

\subsection{Average-Case Setting}
The following two lemmas are useful for computing $T_{i,j}$ recursively for different values of $i$ and $j$. (The proofs of all lemmas can be found in the appendix.)  

\begin{lemma}\label{lem:Tii}
$T_{0,0}=0$, $T_{1,1}=\frac{3}{2}$, and for all $1< i< l$, \[T_{i,i}=\frac{3}{4}T_{i-1,i-1}+\frac{1}{4}T_{i-2,i-1}+\frac{3}{2}.\]
\end{lemma}

\begin{lemma}\label{lem:Tij}
For all $1\leq j<l$, ${T_{0,j}=j}$; for all $1< j<l$, $T_{1,j}=j+\frac{1}{4}$; and for all $1< i\leq l-1$ and $1-i< j\leq l-i$, \[T_{i,i+j} = \frac{3}{4}T_{i-1,i+j-1}+\frac{1}{8}T_{i-2,i+j-1}+\frac{1}{8}T_{i+j-2,i+j-2}+\frac{3}{2}.\]	
\end{lemma}

For any $0\leq \Delta\leq n-1$, define $\Delta_n = \frac{n}{2}-\left|\Delta-\frac{n}{2}\right|$. For simplifying the notation, let
\[
q_{\Delta,i} \triangleq 
\begin{cases}
\frac{2^{i-1}\Delta_n}{2^{l-1}-2^{i-1}\Delta_n}, & \quad \Delta_n < 2^{l-(i+1)}, 2^{l-(i+1)} \geq 1,\\
1, & \quad \Delta_n \geq 2^{l-(i+1)}, 2^{l-(i+1)} \geq 1,\\
0, &\quad \text{otherwise},
\end{cases}
\] for all $1\leq i<l$, and 
\[   
q_{\Delta,0} \triangleq 
\begin{cases}
\frac{\Delta_n}{2^{l-1}}, & \quad \Delta_n < 2^{l-1},\\
1, &\quad \text{otherwise}.\\
\end{cases}
\] Also, let $q_{\Delta,l}\triangleq 1$. Moreover, let
\[
p_{\Delta,j} \triangleq 
\begin{cases}
\frac{2^{l-1}-2^j\Delta_n}{2^{l-1}-2^{j-1}\Delta_n}, & \quad \Delta_n < 2^{l-1}, 2^{l-1}\geq 1,\\
0, &\quad \text{otherwise},\\
\end{cases}
\] for all $1\leq j<l$, and 
\[
p_{\Delta,0} \triangleq 
\begin{cases}
\frac{2^{l-1}-\Delta_n}{2^{l-1}} &\quad \Delta_n < 2^{l-1}, 2^{l-1}\geq 1\\
0 & \quad \text{otherwise}.\\
\end{cases}
\] 

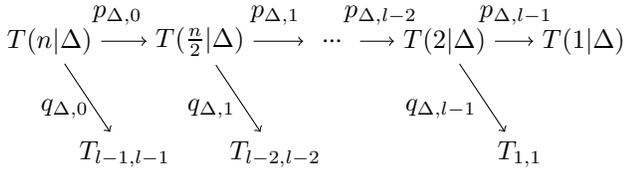
\begin{figure}
\centering
\begin{tikzpicture}[node distance=2cm]
\node (A) at (0, 0) {$T(n\vert\Delta)$};
\node (B) at (2, 0) {$T(\frac{n}{2}\vert\Delta)$};
\node (C) at (1,-1.5) {$T_{l-1,l-1}$};
\node (E) at (3,-1.5) {$T_{l-2,l-2}$};
\node (G) at (3.75,0) {~...~~};
\node (H) at (5.25,0) {$T(2\vert\Delta)$};
\node (I) at (7.1,0) {$T(1\vert\Delta)$};
\node (J) at (6.25,-1.5) {$T_{1,1}$};
\node (K) at (0.9, 0.35) {$p_{\Delta,0}$};
\node (L) at (0.2, -0.9) {$q_{\Delta,0}$};
\node (M) at (3, 0.35) {$p_{\Delta,1}$};
\node (N) at (2.15, -0.9) {$q_{\Delta,1}$};
\node (O) at (4.4, 0.35) {$p_{\Delta,l-2}$};
\node (P) at (5.2, -0.9) {$q_{\Delta,l-1}$};
\node (Q) at (6.2, 0.35) {$p_{\Delta,l-1}$};
\draw[->]
  (A) edge (C) (A) edge (B) (B) edge (G) (B) edge (E) (G) edge (H) (H) edge (I) (H) edge (J);
\end{tikzpicture}
\caption{Recursive form of $T(n\vert\Delta)$.} \label{fig:fig2}
\end{figure}

The following lemma is useful for computing $T(n|\Delta)$ based on the values of $T_{i,j}$. 

\begin{lemma}\label{lem:TDelta}
For any $0\leq \Delta\leq n-1$, we have
\begin{equation*}\label{eq:TDelta}
T(n|\Delta)=\sum_{i=0}^{l-1}m_{\Delta,i}(T_{l-i-1,l-i-1}+i+1)+m_{\Delta,l}l
\end{equation*}
where $m_{\Delta,0}=q_{\Delta,0}$, and \[m_{\Delta,i}=q_{\Delta,i}\prod_{j=0}^{i-1}p_{\Delta,j}\] for all $1\leq i\leq l$.
\end{lemma}

By combining the results of Lemmas~\ref{lem:Tii}--\ref{lem:TDelta}, we can compute $T_{\mathrm{ave}}(n)$ for the proposed strategy as follows. 

\begin{theorem}
For the proposed strategy, $T_{\mathrm{ave}}(n)$ can be computed as $T_{\mathrm{ave}}(n)=P M I$, where $P = [P_0,\dots,P_{n-1}]$ is a row vector of length $n$, where $P_{\Delta}=\frac{2(n-\Delta)}{n(n+1)}$ for all $0\leq \Delta\leq n-1$; and $I = [I_1,\dots,I_{l},l]^{\top}$ is a column vector of length $l+1$, where $I_i = T_{l-i,l-i}+i$ for all $1\leq i\leq l$; and ${M = (m_{\Delta,i})_{0\leq \Delta\leq n-1, 0\leq i\leq l}}$ is an $n\times (l+1)$ matrix, where $\{m_{\Delta,i}\}$ are defined in Lemma~\ref{lem:TDelta}.
\end{theorem}

\begin{proof}
Fix an arbitrary $0\leq \Delta\leq n-1$. It is easy to verify that there exist $n-\Delta$ distinct configurations $\{w_i\}$ such that $\Delta(\{w_i\})=\Delta$. Also, the total number of possible configurations are $n+\binom{n}{2}=\frac{n(n+1)}{2}$. Thus, for a randomly chosen configuration $\{w_i\}$, the probability that $\Delta(\{w_i\})=\Delta$ is equal to $P_{\Delta}=\frac{2(n-\Delta)}{n(n+1)}$. Then, it is easy to see that $T_{\mathrm{ave}}(n)=\sum_{\Delta=0}^{n-1} P_{\Delta}T(n|\Delta)$. Re-writing this equation in matrix form by using the result of Lemma~\ref{lem:TDelta}, the result of the theorem follows immediately.
\end{proof}

\subsection{Worst-Case Setting}

\begin{theorem}
For the proposed strategy, we have ${T_{\mathrm{max}}(n)=2\log_2 n-1}$.
\end{theorem}

\begin{proof}
First, we prove that $T_{\mathrm{max}}(n)=T^{\star}_{2}(n)={T^{\star}_{2}(\frac{n}{2})+2}$. It is easy to verify that $T^{\star}(n|\Delta)$ and $T^{\star}_{i,i}$ can be computed recursively similar to $T(n|\Delta)$ and $T_{i,i}$, respectively, as shown in Fig.~\ref{fig:fig2} and Fig.~\ref{fig:fig3}, by replacing $T$ with $T^{\star}$ everywhere. As can be seen in Fig. \ref{fig:fig2}, the straight lines correspond to the cases in which one weighing resolves the weights of half of the coins; whereas, the diagonal lines correspond to the cases in which the weight of none of the coins is determined. That is, the diagonal lines correspond to the cases that require more number of weighings. Moreover, from $T^{\star}(n|\Delta)$ to $T^{\star}(\frac{n}{4},\frac{n}{4})$, there are two ways (see Fig. \ref{fig:fig3}); one way is through $T^{\star}(\frac{n}{2}|\Delta)$ which requires two weighings, and the other way is through $T^{\star}(\frac{n}{2},\frac{n}{2})$ which requires, in worst case, three weighings, noting that $T^{\star}(\frac{n}{2},\frac{n}{2})=T^{\star}(\frac{n}{4},\frac{n}{4})+2$. Thus, among the diagonal lines, the first one, reaching to $T^{\star}(\frac{n}{2},\frac{n}{2})$, yields the maximum number of required weighings. By these arguments, $T_2^{\star}(n)=T^{\star}(\frac{n}{2},\frac{n}{2})+1$. Similarly, it can be shown that $T_2^{\star}(\frac{n}{2})=T^{\star}(\frac{n}{4},\frac{n}{4})+1$. Thus, $T_2^{\star}(n)=T_2^{\star}(\frac{n}{2})+2$. More generally, we can write the recursive formula $T_2^{\star}(2^i)=T_2^{\star}(2^{i-1})+2$ for all $1<i\leq l$. Noting that $T_2^{\star}(2)=1$, by solving the above recursion, we have $T^{\star}_2(n)=2\log_2 n-1$.    
\end{proof}

\begin{figure}
\vspace{-0.15cm}
\centering
\begin{tikzpicture}[node distance=2cm]
\node (A) at (0, 0) {$T_{i,i}$};
\node (B) at (2.25, 0) {$$};
\node (C) at (1,-1.5) {$T_{i-1,i-1}$};
\node (E) at (3.5,-1.5) {$T_{i-1,i-1}$};
\node (G) at (4.6,0) {$T_{i-2,i-1}$};
\node (K) at (1.2, 0.35) {$\frac{1}{2}$};
\node (L) at (0.3, -0.85) {$\frac{1}{2}$};
\node (M) at (3.2, 0.35) {$\frac{1}{2}$};
\node (N) at (2.7, -0.85) {$\frac{1}{2}$};
\draw[->]
  (A) edge (C) (A) edge (B) (B) edge (G) (B) edge (E) ;
\end{tikzpicture}
\caption{ Recursive form of $T_{i,i}$.} \label{fig:fig3}
\end{figure}
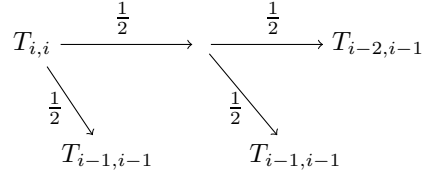

\section{Optimal Nested Weighing Strategy}\label{sec:ONWS}
In a nested strategy, followed by weighing a subset $S$ of coins, if the weight of some coin(s) in $S$ remains undetermined, the next weighing must be performed on a proper subset of $S$. Moreover, if there are multiple such subsets $S$, this procedure must be performed separately for each $S$.  

\subsection{Average-Case Setting}
For any collection $S$ of coins, denote by $d(S)$ the number of coins in $S$ with non-zero weight. For any $1\leq s\leq n$, $w\in \{1,2\}$, and $d\in \{1,2\}$, denote by $\Psi_d(s,w)$ an optimal nested strategy for all collections $S$ of coins, each with an unknown weight in the set $\{0,1,2\}$, such that $|S|=s$, $w(S)= w$, and $d(S)=d$. That is, the expected number of weighings required by the strategy $\Psi_d(s,w)$ over all such $S$ (for any given $s$, $w$, and $d$) is minimum, among all possible nested strategies. Similarly, define $\Psi(s,w)$ as $\Psi_d(s,w)$, except when the expectation is taken over all $S$ such that $|S|=s$ and $w(S)=w$, and define the strategy $\Psi$ as $\{\Psi(s,w)\}_{1\leq s\leq n, 1\leq w\leq 2}$. We wish to design the strategy $\Psi$ and analyze $T^{\Psi}_{\mathrm{ave}}(n)$.


Take an arbitrary collection $S$ of coins such that $|S|=s$, $w(S)=w$, and $d(S)=d$. Consider the application of a nested strategy, represented by $\Psi^{m}_d(s,w)$, on $S$ as follows. The strategy $\Psi^{m}_d(s,w)$ begins with weighing an arbitrary subset $R$ of coins in $S$ of size $1\leq m\leq |S|-1$. If $w(R)=0$ or $w(R)=2$, the strategy $\Psi^{m}_d(s,w)$ proceeds with applying the strategy $\Psi_d(s-m,w)$ on $S\setminus R$, or the strategy $\Psi_d(m,w)$ on $R$, respectively. Otherwise, the strategy $\Psi^{m}_d(s,w)$ applies the strategies $\Psi_d(m,1)$ and $\Psi_d(s-m,1)$ on $R$ and $S\setminus R$, respectively. Denote by $T^{m}_{w,d}(s)$ the expected number of weighings required by the strategy $\Psi^m_d(s,w)$ over all such $S$, and let $T^{\mathrm{opt}}_{w,d}(s)\triangleq \min_{1\leq m\leq s-1} T^{m}_{w,d}(s)$. Similarly, define the strategy $\Psi^m(s,w)$ the same as $\Psi^m_d(s,w)$, except when $\Psi_d$ is replaced by $\Psi$ everywhere. Denote by $T^{m}_{w}(s)$ the expected number of weighings required by the strategy $\Psi^m(s,w)$ over all $S$ such that $|S|=s$ and $w(S)=w$, and let $T^{\mathrm{opt}}_{w}(s)\triangleq \min_{1\leq m\leq s-1} T^{m}_{w}(s)$. A simple recursive argument yields that for the strategy $\Psi$ defined earlier, we have $T^{\Psi}_{\mathrm{ave}}(n) = T^{\mathrm{opt}}_2(n)$. 


For the ease of notation, for any $2\leq s\leq n$ and ${1\leq m\leq s-1}$, we define $\alpha_{i,j}(s,m)\triangleq {\binom{s-i}{m-j}}/{\binom{s}{m}}$ for all $i,j$ such that $0\leq m-j\leq s-i$, and define $\alpha_{i,j}(s,m)\triangleq 0$, otherwise. For brevity, we simply refer to $\alpha_{i,j}(s,m)$ by $\alpha_{i,j}$ whenever $s$ and $m$ are clear from the context. 

Based on the above definitions, the following results can be shown. 

\begin{lemma}\label{lem:Tmw}
For any $2\leq s\leq n$ and $1\leq m\leq s-1$, we have 
\begin{dmath*}\label{eq:ET1}
T^{m}_{1}(s) =  \alpha_{1,0} (T^{\mathrm{opt}}_{1}(s-m)+1)+\alpha_{1,1} (T^{\mathrm{opt}}_{1}(m)+1),
\end{dmath*} where $T^{\mathrm{opt}}_1(1)=0$. Moreover, for any $3\leq s\leq n$ and $1\leq m\leq s-1$, we have
\begin{dmath*}\label{eq:ET2}
T^{m}_2(s) = \frac{2}{s+1}T^{m}_{2,1}(s)+\frac{s-1}{s+1} T^{m}_{2,2}(s),
\end{dmath*} 
\begin{dmath*}\label{eq:ET21}
T^{m}_{2,1}(s) = \alpha_{1,0}(T^{\mathrm{opt}}_{2,1}(s-m)+1)+\alpha_{1,1}(T^{\mathrm{opt}}_{2,1}(m)+1), 	
\end{dmath*} and
\begin{dmath*}\label{eq:ET22}
T^{m}_{2,2}(s) = \alpha_{2,0}(T^{\mathrm{opt}}_{2,2}(s-m)+1)+\alpha_{2,2}(T^{\mathrm{opt}}_{2,2}(m)+1)+2\alpha_{2,1}(T^{\mathrm{opt}}_{1}(m)+T^{\mathrm{opt}}_1(s-m)+1), 	
\end{dmath*} where $T^{\mathrm{opt}}_{2,1}(1) = T^{\mathrm{opt}}_{2,2}(1)=0$, and $T^{\mathrm{opt}}_{2,1}(2)=T^{\mathrm{opt}}_{2,2}(2)=1$. 
\end{lemma}



\begin{lemma}\label{lem:MidPoint}
For any $2\leq s\leq n$, we have $\lfloor \frac{s}{2}\rfloor,\lceil\frac{s}{2}\rceil\in \argmin_{1\leq m\leq s-1} T^{m}_1(s)$; for any $3\leq s\leq n$ and $d\in\{1,2\}$, we have $\lfloor \frac{s}{2}\rfloor,\lceil\frac{s}{2}\rceil\in \argmin_{1\leq m\leq s-1} T^{m}_{2,d}(s)$; and for any $3\leq s\leq n$, we have $\lfloor \frac{s}{2}\rfloor,\lceil\frac{s}{2}\rceil\in \argmin_{1\leq m\leq s-1} T^{m}_{2}(s)$.	
\end{lemma}

\begin{lemma}\label{average-lemma-2}
For any $0\leq i\leq l$, we have 
$T^{\mathrm{opt}}_1(2^i)=i$, and $T^{\mathrm{opt}}_2(2^i) = \frac{(i-1)2^{i+1}+i+2}{2^i+1}$.
\end{lemma}


Recall that for the optimal nested strategy $\Psi$ defined earlier, we have $T^{\Psi}_{\mathrm{ave}}(n) = T^{\mathrm{opt}}_2(n)$. Thus the following result is immediate by the result of Lemma~\ref{average-lemma-2}. 

\begin{theorem}\label{averagetheorem}
For the optimal nested strategy $\Psi$, we have $T_{\mathrm{ave}}(n)=\frac{2n+1}{n+1}\log_2 n-\frac{2(n-1)}{n+1}$.
\end{theorem}

%
%

\subsection{Worst-Case Setting}
Consider an optimal nested strategy $\Psi^{\star}$ for the worst-case setting, defined similarly as the strategy $\Psi$ for the average-case setting, except when considering the maximum number of required weighings (instead of the expected number of required weighings). Then, the following result holds~\cite{WZC:17}.


\begin{theorem}\cite{WZC:17}\label{worsttheorem}
For the optimal nested strategy $\Psi^{\star}$, we have $T_{\mathrm{max}}(n)=2\log_2 n-1$.
\end{theorem}


\begin{figure}
\centering
\includegraphics[width=0.45\textwidth]{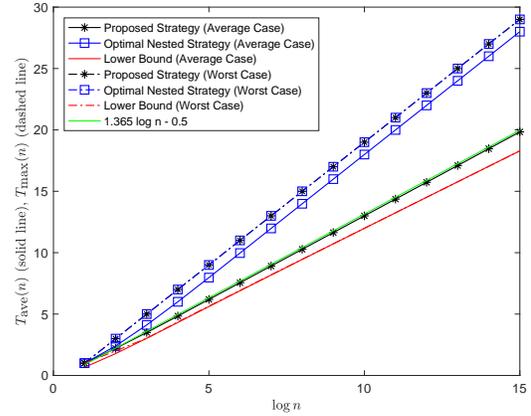}\vspace{-0.25cm}	
\caption{The average-case and worst-case results for the proposed strategy, the optimal nested strategy, and the information-theoretic lower bound.}\label{fig:Results}\vspace{-0.25cm}
\end{figure}

\section{Comparison Results}
In this section, we present our numerical results for the performance of the proposed strategy in both the average-case and worst-case settings. For each setting, the performance of the proposed strategy is compared with the performance of the optimal nested strategy (defined in Section~\ref{sec:ONWS}) and the information-theoretic lower bound (Theorem~\ref{lem:bound}). 

Fig.~\ref{fig:Results} illustrates that the proposed strategy, in the average-case setting, significantly outperforms the optimal nested strategy. Also, in the worst-case setting, the proposed strategy achieves the same performance as  the nested strategy. Our numerical evaluations suggest that the expected number of weighings required by the proposed strategy, which is computable using the recursive formulas in Section~\ref{sec:Analysis}, can be also approximated by $1.365\log_2 n-0.5$ as $n$ grows unbounded (see Fig.~\ref{fig:Results}). In this asymptotic regime, the optimal nested strategy requires $2\log_2 n-2$ weighings on average, and the information-theoretic lower bound is $2\log_3 n\approx 1.262\log_2 n$. Thus, a simple calculation shows that the proposed strategy, when compared to the optimal nested strategy, requires about $31.75\%$ less number of weighings on average. Additionally, when compared to the information-theoretic lower bound, the proposed strategy requires at most about $8.16\%$ extra number of weighings on average.   

\bibliographystyle{IEEEtran}
\bibliography{QGTRefs}

\begin{thebibliography}{10}
\providecommand{\url}[1]{#1}
\csname url@samestyle\endcsname
\providecommand{\newblock}{\relax}
\providecommand{\bibinfo}[2]{#2}
\providecommand{\BIBentrySTDinterwordspacing}{\spaceskip=0pt\relax}
\providecommand{\BIBentryALTinterwordstretchfactor}{4}
\providecommand{\BIBentryALTinterwordspacing}{\spaceskip=\fontdimen2\font plus
\BIBentryALTinterwordstretchfactor\fontdimen3\font minus
  \fontdimen4\font\relax}
\providecommand{\BIBforeignlanguage}[2]{{%
\expandafter\ifx\csname l@#1\endcsname\relax
\typeout{** WARNING: IEEEtran.bst: No hyphenation pattern has been}%
\typeout{** loaded for the language `#1'. Using the pattern for}%
\typeout{** the default language instead.}%
\else
\language=\csname l@#1\endcsname
\fi
#2}}
\providecommand{\BIBdecl}{\relax}
\BIBdecl

\bibitem{B:09}
N.~H. Bshouty, ``{Optimal Algorithms for the Coin Weighing Problem with a
  Spring Scale},'' in \emph{Conference on Learning Theory}, 2009.

\bibitem{DH:2000}
D.~Du and F.~Hwang, \emph{Combinatorial {G}roup {T}esting and {I}ts
  {A}pplications}, ser. Applied Mathematics.\hskip 1em plus 0.5em minus
  0.4em\relax World Scientific, 2000.

\bibitem{KKLP:17}
S.~Keiper, G.~Kutyniok, D.~G. Lee, and G.~E. Pfander, ``{Compressed Sensing for
  Finite-Valued Signals},'' \emph{Linear Algebra and its Applications}, vol.
  532, pp. 570--613, 2017.

\bibitem{WZC:17}
C.~Wang, Q.~Zhao, and C.~N. Chuah, ``{Optimal Nested Test Plan for
  Combinatorial Quantitative Group Testing},'' \emph{IEEE Transactions on
  Signal Processing}, vol.~PP, no.~99, 2017.

\bibitem{S:60}
H.~S. Shapiro, ``{Problem E 1399},'' \emph{Amer.~Math.~Monthly}, vol.~67,
  no.~82, pp. 697--697, 1960.

\bibitem{L:75}
B.~Lindstr{\"o}m, ``{Determining Subsets by Unramified Experiments},'' in
  \emph{A survey of Statistical Design and Linear Models}.\hskip 1em plus 0.5em
  minus 0.4em\relax North-Holland, 1975.

\bibitem{AS:85}
M.~Aigner and M.~Schughart, ``{Determining Defectives in a Linear Order},''
  \emph{Journal of Statistical Planning and Inference}, vol.~12, pp. 359--368,
  1985.

\bibitem{A:86}
M.~Aigner, ``{Search Problems on Graphs},'' \emph{Discrete Applied
  Mathematics}, vol.~14, no.~3, pp. 215--230, 1986.

\bibitem{MK:90}
S.~Martirosyan and G.~Khachatryan, ``{Construction of Signature Codes and the
  Coin Weighing Problem},'' vol.~25, 04 1990.

\bibitem{A:88}
M.~Aigner, \emph{Combinatorial Search}.\hskip 1em plus 0.5em minus 0.4em\relax
  New York, NY, USA: John Wiley \&amp; Sons, Inc., 1988.

\bibitem{UTW:2000}
R.~Uehara, K.~Tsuchida, and I.~Wegener, ``{Identification of Partial
  Disjunction, Parity, and Threshold Functions},'' \emph{Theoretical Computer
  Science}, vol. 230, no.~1, pp. 131--147, 2000.

\bibitem{EM:14}
A.~Emad and O.~Milenkovic, ``{Semiquantitative Group Testing},'' \emph{IEEE
  Transactions on Information Theory}, vol.~60, no.~8, pp. 4614--4636, Aug.
  2014.

\bibitem{ARP:17}
P.~Abdalla, A.~Reisizadeh, and R.~Pedarsani, ``{Multilevel Group Testing via
  Sparse-Graph Codes},'' in \emph{IEEE 51st Asilomar Conference on Signals,
  Systems, and Computers}, Oct. 2017, pp. 895--899.

\end{thebibliography}

\appendix[Proofs of Lemmas]

\begin{proof}[Proof of Lemma~\ref{lem:Tii}]
It is easy to see that $T_{0,0}=T(1,1)=0$, since we have two coins and the weight of each coin is $1$, so no weighing is required. To obtain $T_{1,1}=T(2,2)$, we know that we have four coins and the total weight of two coins (set $A$) is $1$ and the total weight of the other two (set $B$) is also $1$. Thus, one coin in $A$ and one coin in $B$ are weighed together. The weighing outcome is either (i) $0$ or $2$ with probability $\frac{1}{2}$, or (ii) $1$ with probability $\frac{1}{2}$. In the case (i), with just one weighing the weights of all coins are determined. In the case (ii), one more weighing is needed to be performed on one coin in $A$ or one coin in $B$, in order to find the weights of all coins. Thus, $T_{1,1}=\frac{1}{2}(1)+\frac{1}{2}(2)=\frac{3}{2}$. 

For any $1<i<l$, $T_{i,i}$ can be computed based on a similar reasoning as follows. By performing one weighing on the union set of half of $A$ (say $A_1$) and half of $B$ (say $B_1$), with probability $\frac{1}{2}$ the weighing outcome is $0$ or $2$, and the expected number of extra required weighings is $T_{i-1,i-1}$. Otherwise, with probability $\frac{1}{2}$, the weighing outcome is $1$. In this case, one more weighing is needed to be performed on the union set of $A_1$ and half of $B_2$ (say $(B_2)_1$). Followed by two weighings, with probability $\frac{1}{2}$, the expected number of extra required weighings is equal to $T_{i-1,i-1}$; and with probability $\frac{1}{2}$, this quantity is equal to $T_{i-2,i-1}$ (see Fig. \ref{fig:fig3}). Thus, we have \[T_{i,i}=\frac{1}{2}(T_{i-1,i-1}+1)+\frac{1}{4}(T_{i-1,i-1}+2)+\frac{1}{4}(T_{i-2,i-1}+2),\] or equivalently, \[T_{i,i}=\frac{3}{4}T_{i-1,i-1}+\frac{1}{4}T_{i-2,i-1}+\frac{3}{2}.\] This completes the proof.
\end{proof}

\begin{proof}[Proof of Lemma~\ref{lem:Tij}]
For any $T_{i,j}$, we consider two disjoint sets $A$ and $B$ of size $2^i$ and $2^j$, respectively, each set of total weight $1$. By the definition, $T_{0,j}=T(1,2^j)=T_1(2^j)$. In this case, the proposed strategy applies the procedure $\Pi_0(B)$, which requires $j$ weighings, on average, to determine the weights of all coins. Thus, $T_{0,j}=j$.   
	
Now, consider $T_{1,j}=T(2,2^j)$. In this case, with one weighing (on the set $A_1\cup B_1$) with probability $\frac{1}{2}$, the outcome is $0$ or $2$. For the outcome $0$ (or $2$), we find that one coin with weight $1$ is in $A_2$ (or $A_1$) and the other coin of weight $1$ is in $B_2$ (or $B_1$). Thus, $T(1,2^{j-1})=j-1$ more weighings, on average, are needed to determine the weights of all coins. Also, with probability $\frac{1}{2}$, the weighing outcome is $1$ and the weight of no coin is discovered by this particular weighing. In this case, with one more weighing (on the set $A_1\cup (B_2)_1$), with probability $\frac{1}{2}$ the weighing outcome is $0$, and $T(1,2^{j-1})=j-1$ more weighings, on average, are needed to determine the weights of all coins. Otherwise, with probability $\frac{1}{2}$, the weighing outcome is $1$ or $2$, and thus, $T(1,2^{j-2})=j-2$ more weighings are needed on average to find the weights of all coins. As a result, we have 
\begin{dmath*}
T_{1,j}=\frac{1}{2}(j-1+1)+\frac{1}{4}(j-1+2)+\frac{1}{4}(j-2+2)=j+\frac{1}{4}.	
\end{dmath*} 
	
Lastly, consider $T_{i,i+j}=T(2^{i},2^{i+j})$. Performing one weighing (on the set $A_1\cup B_1$), there are two cases: (i) the outcome is $0$ or $2$ (with probability $\frac{1}{2}$), and (ii) the outcome is $1$ (with probability $\frac{1}{2}$). First, consider the case (i). By a similar argument as before, for the outcome $0$ (or $2$), we find the weight of all coins in $A_1$ (or $A_2$) and that of all coins in $B_1$ (or $B_2$), respectively. Thus, $T(2^{i-1},2^{i+j-1})$ more weighings are needed, on average, to determine the weights of all coins. Next, consider the case (ii). There are two sub-cases: (ii-1) with probability $\frac{1}{4}$, the coin with weight $1$ belongs to the larger set, and (ii-2) with probability $\frac{1}{4}$, the coin with weight $1$ belongs to the smaller set. In the case (ii-1), the weight of no coin can be determined by this particular weighing. Thus, with one more weighing (on the set $(A_2)_1\cup B_1$ or $A_1\cup (B_2)_1$), with probability $\frac{1}{2}$ the weighing outcome is $0$, and consequently, $T(2^{i-1},2^{i+j-1})$ more weighings are needed, on average, to determine the weights of all coins; otherwise, with probability $\frac{1}{2}$, the weighing outcome is $1$ or $2$, and $T(2^{i-2},2^{i+j-1})$ more weighings are needed, on average, to find the weights of all coins. In the case (ii-2), again, the weight of no coin is determined by this specific weighing. Thus, with one more weighing (on set $(A_2)_1\cup B_1$ or $A_1\cup (B_2)_1$), with probability $\frac{1}{2}$, the weighing outcome is $0$, and $T(2^{i-1},2^{i+j-1})$ more weighings are needed, on average, to determine the weights of all coins. Otherwise, with probability $\frac{1}{2}$, the weighing outcome is $1$ or $2$, and $T(2^{i+j-2},2^{i+j-2})$ more weighings, on average, are needed to find the weights of all coins. Thus, we have 
\begin{dmath*}
T_{i,i+j}=\frac{1}{2}(T_{i-1,i+j-1}+1)+\frac{1}{8}(T_{i-1,i+j-1}+2)+\frac{1}{8}(T_{i-2,i+j-1}+2)
	+\frac{1}{8}(T_{i-1,i+j-1}+2)+\frac{1}{8}(T_{i+j-2,i+j-2}+2),	
\end{dmath*} or in turn, 
\begin{dmath*}
T_{i,i+j}=\frac{3}{4}T_{i-1,i+j-1}+\frac{1}{8}T_{i-2,i+j-1}+\frac{1}{8}T_{i+j-2,i+j-2}+\frac{3}{2}.	
\end{dmath*} This completes the proof.
\end{proof}

\begin{proof}[Proof of Lemma~\ref{lem:TDelta}]
Fix an arbitrary $0\leq \Delta\leq n-1$. Consider the application of the proposed strategy on an arbitrary configuration $\{w_i\}$ such that $\Delta(\{w_i\})=\Delta$. Let $p_{\Delta,0}$ (or respectively, $q_{\Delta,0}$) be the probability that the outcome of the first weighing is $0$ or $2$ (or respectively, $1$). Thus, with probability $p_{\Delta,0}$, followed by performing one weighing, the expected number of extra required weighings is $T(\frac{n}{2}|\Delta)$. Similarly, with probability $p_{\Delta,0} p_{\Delta,1}$, after performing two weighings, $T(\frac{n}{4}|\Delta)$ more weighings are needed on average, and so forth (see the straight line in Fig.~\ref{fig:fig2}). Thus, with probability $\prod_{i=0}^{l-1} p_{\Delta,i}$, $l$ weighings are needed. On the other hand, with probability $q_{\Delta,0}$, followed by performing one weighing, $T_{l-1,l-1}$ more weighings are needed on average. Similarly, with probability $p_{\Delta,0} q_{\Delta,1}$, after performing two weighings, $T_{l-2,l-2}$ extra number of weighings are needed on average, and so forth (see the diagonal lines in Fig.~\ref{fig:fig2}). Putting everything together, we can write 
\begin{dmath*}
T(n|\Delta) = q_{\Delta,0}(T_{l-1,l-1}+1)+p_{\Delta,0}q_{\Delta,1}(T_{l-2,l-2}+2)+\dots+p_{\Delta,0}p_{\Delta,1}\cdots p_{\Delta,l-1}(l),
\end{dmath*} or equivalently, 
\begin{dmath*}
T(n|\Delta)=m_{\Delta,0}(T_{l-1,l-1}+1)+m_{\Delta,1}(T_{l-2,l-2}+2)+\dots+m_{\Delta,l-1}(T_{0,0}+l)+m_{\Delta,l}(l).\end{dmath*} This completes the proof.
\end{proof}

\begin{proof}[Proof of Lemma~\ref{lem:Tmw}]
Consider an arbitrary collection of $s$ coins of total weight $w$. There are two cases: (i) $w=1$, and (ii) $w=2$. In the case (i), a randomly chosen subset of $m$ coins weighs $0$ with probability $\alpha_{1,0}$, and it weighs $1$ with probability $\alpha_{1,1}$. In these two sub-cases the expected number of extra required weighings is $T^{\mathrm{opt}}_{1,1}(s-m)$ and $T^{\mathrm{opt}}_{1,1}(m)$, respectively. Note that $T^m_{1,1}(t) = T^{m}_{1}(t)$ for all $t$, and so, $T^{\mathrm{opt}}_{1,1}(t) = T^{\mathrm{opt}}_{1}(t)$ for all $t$. Thus, \begin{dmath*}
T^m_1(s) = \alpha_{1,0} (T^{\mathrm{opt}}_{1}(s-m)+1)+\alpha_{1,1} (T^{\mathrm{opt}}_{1}(m)+1)
\end{dmath*} 
In the case (ii), there are two sub-cases: (ii-1) there is one coin of weight $2$ (there exist $s$ distinct sub-configurations with this characteristic), and (ii-2) there are two coins, each of weight $1$ (there exist $\binom{s}{2}$ distinct such sub-configurations). Thus, \[T^{m}_2(s) = \frac{2}{s+1}T^{m}_{2,1}(s)+\frac{s-1}{s+1} T^{m}_{2,2}(s).\] In the case (ii-1), a randomly chosen subset of $m$ coins weighs $0$ or $2$ with probability $\alpha_{1,0}$ or $\alpha_{1,1}$, respectively, and in these two sub-cases the expected number of extra required weighings is $T^{\mathrm{opt}}_{2,1}(s-m)$ and $T^{\mathrm{opt}}_{2,1}(m)$, respectively. Thus, \begin{dmath*}
T^{m}_{2,1}(s) = \alpha_{1,0}(T^{\mathrm{opt}}_{2,1}(s-m)+1)+\alpha_{1,1}(T^{\mathrm{opt}}_{2,1}(m)+1). 	
\end{dmath*} Similarly, in the case (ii-2), a randomly chosen subset of $m$ coins weighs $0$, or $1$, or $2$ with probability $\alpha_{2,0}$, or $2\alpha_{2,1}$, or $\alpha_{2,2}$, respectively. In these three sub-cases, the expected number of extra required weighings is $T^{\mathrm{opt}}_{2,2}(s-m)$, $T^{\mathrm{opt}}_{1}(m)+T^{\mathrm{opt}}_{1}(s-m)$, and $T^{\mathrm{opt}}_{2,2}(m)$, respectively. Thus, 
\begin{dmath*}
T^{m}_{2,2}(s) = \alpha_{2,0}(T^{\mathrm{opt}}_{2,2}(s-m)+1)+2\alpha_{2,1}(T^{\mathrm{opt}}_{1}(m)+T^{\mathrm{opt}}_1(s-m)+1)+\alpha_{2,2}(T^{\mathrm{opt}}_{2,2}(m)+1). 	
\end{dmath*} 
This completes the proof.
%
%
%
%
\end{proof}

\begin{proof}[Proof of Lemma~\ref{lem:MidPoint}]
The proof techniques are the same for $T^{m}_1(s)$, $T^{m}_{2,d}(s)$, and $T^{m}_2(s)$, and we only state the proof for $T^{m}_1(s)$ to avoid repetition. In particular, we shall show that for any $2\leq s\leq n$, we have $\lfloor \frac{s}{2}\rfloor,\lceil\frac{s}{2}\rceil\in \argmin_{1\leq m\leq s-1} T^{m}_1(s)$. 

The proof is by induction on $s$. It is easy to see that for $s=2$, we have $\lfloor \frac{2}{2}\rfloor,\lceil\frac{2}{2}\rceil=1=\argmin_{1\leq m\leq 2-1} T^{m}_1(2)$. The induction hypothesis is that for $s\leq l-1$, 
\begin{dmath}\label{eq:IH1}
\left\lfloor \frac{s}{2}\right\rfloor,\left\lceil\frac{s}{2}\right\rceil\in \argmin_{1\leq m\leq s-1} T^{m}_1(s)	
\end{dmath} holds. To complete the proof, it is enough to show that for $s=l$, we have $\lfloor \frac{l}{2}\rfloor,\lceil\frac{l}{2}\rceil\in \argmin_{1\leq m\leq l-1} T^{m}_1(l)$. Based on the formula for $T^{m}_1(s)$ in Lemma~\ref{lem:Tmw}, it can be readily confirmed that $T^{m}_1(s)$ is symmetric around midrange point, i.e., $T^{\lceil\frac{s}{2}\rceil}_1(s) = T^{\lfloor\frac{s}{2}\rfloor}_1(s)$. Thus, showing the proof for $\lfloor \frac{s}{2}\rfloor$ suffices. More specifically, we need to show the following:

\begin{dmath}\label{eq:E1}
\frac{\binom{l-1}{\left\lfloor \frac{l}{2}\right\rfloor}}{\binom{l}{\left\lfloor \frac{l}{2}\right\rfloor}}
T^{\mathrm{opt}}_{1}(l-\left\lfloor \frac{l}{2}\right\rfloor)+ \frac{\binom{l-1}{\left\lfloor \frac{l}{2}\right\rfloor-1}}{\binom{l}{\left\lfloor \frac{l}{2}\right\rfloor}}
T^{\mathrm{opt}}_{1}(\left\lfloor \frac{l}{2}\right\rfloor) \leq \frac{\binom{l-1}{m}}{\binom{l}{m}}
T^{\mathrm{opt}}_{1}(l-m) + \frac{\binom{l-1}{m-1}}{\binom{l}{m}}
T^{\mathrm{opt}}_{1}(m) 	
\end{dmath} for all $1\leq m\leq l-1$. Due to the symmetry, it suffices to show that~\eqref{eq:E1} holds for all $1\leq m\leq \lfloor\frac{l}{2}\rfloor$. We consider two cases: (i) $l=2k$, and (ii) $l=2k+1$, for some $k\geq 1$. 

\subsubsection*{Proof for Case (i)} Noting that $l=2k$ and $\lfloor\frac{l}{2}\rfloor=k$,~\eqref{eq:E1} can be written as
\begin{dmath*}
\frac{\binom{2k-1}{k}}{\binom{2k}{k}}
T^{\mathrm{opt}}_{1}(k)+ \frac{\binom{2k-1}{k-1}}{\binom{2k}{k}}
T^{\mathrm{opt}}_{1}(k) \leq \frac{\binom{2k-1}{m}}{\binom{2k}{m}}
T^{\mathrm{opt}}_{1}(2k-m) + \frac{\binom{2k-1}{m-1}}{\binom{2k}{m}}
T^{\mathrm{opt}}_{1}(m),	
\end{dmath*} or equivalently,
\begin{dmath}\label{eq:E3}
T^{\mathrm{opt}}_{1}(k) \leq \left(\frac{2k-m}{2k}\right)
T^{\mathrm{opt}}_{1}(2k-m) + \left(\frac{m}{2k}\right)
T^{\mathrm{opt}}_{1}(m) 	
\end{dmath} for all $1\leq m\leq k$.
An inductive argument (on $m$) is used to prove that~\eqref{eq:E3} holds. It is easy to see that~\eqref{eq:E3} holds for $m=k$, i.e., \begin{dmath*}
 T^{\mathrm{opt}}_{1}(k) = \left(\frac{2k-k}{2k}\right)
T^{\mathrm{opt}}_{1}(2k-k) + \left(\frac{k}{2k}\right)
T^{\mathrm{opt}}_{1}(k).	
\end{dmath*} We assume that~\eqref{eq:E3} holds for $m>t$. We need to show that for $m=t$, the following holds:
\begin{dmath}\label{eq:E4}
T^{\mathrm{opt}}_{1}(k) \leq \left(\frac{2k-t}{2k}\right)
T^{\mathrm{opt}}_{1}(2k-t) + \left(\frac{t}{2k}\right)
T^{\mathrm{opt}}_{1}(t) 	
\end{dmath}. The proof is by contradiction. Suppose that~\eqref{eq:E4} does not hold, i.e., 
\begin{dmath}\label{eq:E5}
T^{\mathrm{opt}}_{1}(k) > \left(\frac{2k-t}{2k}\right)
T^{\mathrm{opt}}_{1}(2k-t) + \left(\frac{t}{2k}\right)
T^{\mathrm{opt}}_{1}(t) 	
\end{dmath}. Since~\eqref{eq:E3} holds for $m>t$ (by assumption), for $m=t+1$ we have
\begin{dmath}\label{eq:E6}
T^{\mathrm{opt}}_{1}(k) \leq \left(\frac{2k-t-1}{2k}\right)
T^{\mathrm{opt}}_{1}(2k-t-1) + \left(\frac{t+1}{2k}\right)
T^{\mathrm{opt}}_{1}(t+1) 	
\end{dmath}. Combining~\eqref{eq:E5} and~\eqref{eq:E6}, we get
\begin{dmath*}\left(\frac{2k-t}{2k}\right)
T^{\mathrm{opt}}_{1}(2k-t) + \left(\frac{t}{2k}\right)
T^{\mathrm{opt}}_{1}(t) 	 < \left(\frac{2k-t-1}{2k}\right)
T^{\mathrm{opt}}_{1}(2k-t-1) + \left(\frac{t+1}{2k}\right)
T^{\mathrm{opt}}_{1}(t+1)
\end{dmath*}
which equivalently can be written as
\begin{dmath}\label{eq:E7}
\left(2k-t\right)
T^{\mathrm{opt}}_{1}(2k-t) -\left(2k-t-1\right)
T^{\mathrm{opt}}_{1}(2k-t-1) < \left(t+1\right)
T^{\mathrm{opt}}_{1}(t+1)-\left(t\right)
T^{\mathrm{opt}}_{1}(t)	
\end{dmath}. We need to disprove~\eqref{eq:E7}. Before moving further with disproving~\eqref{eq:E7}, we shall show the following formula which will be used in the rest of the proof:
\begin{dmath}\label{eq:E8}
T^{\mathrm{opt}}_{1}(q) -
T^{\mathrm{opt}}_{1}(q-1) = \frac{2^{\lfloor \log (q-1) \rfloor+1}}{q(q-1)}
\end{dmath} for all $2\leq q \leq l-1$. The proof of~\eqref{eq:E8} is based on an inductive argument. It is easy to see that for $q=2$, we have $T^{\mathrm{opt}}_{1}(2)-T^{\mathrm{opt}}_{1}(2-1) = 1-0= \frac{2^{\lfloor \log (2-1) \rfloor+1}}{2(2-1)}=1$. We assume that for $q\leq h-1$,~\eqref{eq:E8} holds. It suffices to show that for $q=h$ we have 
\begin{dmath*}
T^{\mathrm{opt}}_{1}(h) -
T^{\mathrm{opt}}_{1}(h-1) = \frac{2^{\lfloor \log (h-1) \rfloor+1}}{h(h-1)}.
\end{dmath*} We consider two cases: (i-1) $h=2d$, and (i-2) $h=2d+1$, for some $d\geq 1$. First, consider the case (i-1). Since $h\leq l-1$ and~\eqref{eq:IH1} holds for $s\leq l-1$, we have $\lfloor \frac{h}{2}\rfloor \in \argmin_{1\leq m\leq h-1} T^{m}_1(h)$ and $\lfloor \frac{h-1}{2}\rfloor \in \argmin_{1\leq m\leq h-2} T^{m}_1(h-1)$.
Using the formula in Lemma~\ref{lem:Tmw} and noting that $h=2d$, $\lfloor \frac{h}{2} \rfloor=d$, and $\lfloor \frac{h-1}{2} \rfloor=d-1$, $T^{\mathrm{opt}}_{1}(h) $ and $T^{\mathrm{opt}}_{1}(h-1)$ can be written as
\begin{dmath}\label{eq:E9}
T^{\mathrm{opt}}_{1}(h)=T^{\mathrm{opt}}_{1}(d)+1,
\end{dmath} and 
\begin{dmath}\label{eq:E10}
T^{\mathrm{opt}}_{1}(h-1) =\left(\frac{d}{2d-1}\right)
T^{\mathrm{opt}}_{1}(d) + \left(\frac{d-1}{2d-1}\right)
T^{\mathrm{opt}}_{1}(d-1) +1.	
\end{dmath}
Subtracting~\eqref{eq:E10} from~\eqref{eq:E9} results in
\begin{dmath}\label{eq:E10-1}
T^{\mathrm{opt}}_{1}(h) -
T^{\mathrm{opt}}_{1}(h-1) = \left(\frac{d-1}{2d-1}\right)(T^{\mathrm{opt}}_{1}(d) -
T^{\mathrm{opt}}_{1}(d-1)).	
\end{dmath}
Since~\eqref{eq:E8} holds for $q\leq h-1$ (by assumption), we have 
\begin{dmath*}
(T^{\mathrm{opt}}_{1}(d) - T^{\mathrm{opt}}_{1}(d-1))= \frac{2^{\lfloor \log (d-1) \rfloor+1}}{d(d-1)}.	
\end{dmath*} Substituting $(T^{\mathrm{opt}}_{1}(d) - T^{\mathrm{opt}}_{1}(d-1))$ by $ \frac{2^{\lfloor \log (d-1) \rfloor+1}}{d(d-1)}$ in~\eqref{eq:E10-1}, we have 
\begin{dmath*}
T^{\mathrm{opt}}_{1}(h) -
T^{\mathrm{opt}}_{1}(h-1) =  \frac{2^{\lfloor \log (d-1) \rfloor+1}}{d(2d-1)}= \frac{2^{\lfloor \log (2d-2) \rfloor+1}}{2d(2d-1)}= \frac{2^{\lfloor \log (2d-1) \rfloor+1}}{2d(2d-1)}= \frac{2^{\lfloor \log (h-1) \rfloor+1}}{h(h-1)}.	
\end{dmath*}
This completes the proof of~\eqref{eq:E8} for the case (i-1). Now, consider the case (i-2). Noting that $h=2d+1$ and using similar arguments as above, it can be shown that 
\begin{dmath*}
T^{\mathrm{opt}}_{1}(h) -
T^{\mathrm{opt}}_{1}(h-1) = \left(\frac{d+1}{2d+1}\right)(T^{\mathrm{opt}}_{1}(d+1) -
T^{\mathrm{opt}}_{1}(d)),	
\end{dmath*}
and 
\begin{dmath*}
(T^{\mathrm{opt}}_{1}(d+1) - T^{\mathrm{opt}}_{1}(d))=\frac{2^{\lfloor \log (d) \rfloor+1}}{d(d+1)}. 	
\end{dmath*} Subsequently, we have 
\begin{dmath*}
T^{\mathrm{opt}}_{1}(h) -
T^{\mathrm{opt}}_{1}(h-1) =  \frac{2^{\lfloor \log (d) \rfloor+1}}{d(2d+1)}= \frac{2^{\lfloor \log (2d) \rfloor+1}}{2d(2d+1)}=\frac{2^{\lfloor \log (h-1) \rfloor+1}}{h(h-1)}.	
\end{dmath*} This completes the proof of~\eqref{eq:E8} for the case (i-2). 

Now, we proceed with disproving~\eqref{eq:E7}. By using~\eqref{eq:E8} and substituting $T^{\mathrm{opt}}_{1}(2k-t)$ by $ T^{\mathrm{opt}}_{1}(2k-t-1)+\frac{2^{\lfloor \log 2k-t-1 \rfloor+1}}{(2k-t-1)(2k-t)}$ and $T^{\mathrm{opt}}_{1}(t+1)$ by $ T^{\mathrm{opt}}_{1}(t)+\frac{2^{\lfloor \log t \rfloor+1}}{t(t+1)}$ in~\eqref{eq:E7}, we have
\begin{dmath*}
T^{\mathrm{opt}}_{1}(2k-t-1)+\frac{2^{\lfloor \log 2k-t-1 \rfloor+1}}{2k-t-1} <  T^{\mathrm{opt}}_{1}(t)+\frac{2^{\lfloor \log t \rfloor+1}}{t}	
\end{dmath*} which equivalently can be written as
\begin{dmath}\label{eq:E11}
T^{\mathrm{opt}}_{1}(2k-t-1)-T^{\mathrm{opt}}_{1}(t)< \frac{2^{\lfloor \log t \rfloor+1}}{t} -\frac{2^{\lfloor \log 2k-t-1 \rfloor+1}}{2k-t-1}	
\end{dmath}. To complete disproving~\eqref{eq:E7}, we need to show that~\eqref{eq:E11} does not hold. To this end, we need to prove the following: 
\begin{dmath}\label{eq:E12}
T^{\mathrm{opt}}_{1}(a+i)-T^{\mathrm{opt}}_{1}(a)\geq \frac{2^{\lfloor \log a \rfloor+1}}{a} -\frac{2^{\lfloor \log a+i \rfloor+1}}{a+i}	
\end{dmath} for all $1\leq a \leq l-2$ and $1\leq i\leq l-1-a$. The proof of~\eqref{eq:E12} is based on an inductive argument (on $i$). For $i=1$, we need to show that 
\begin{dmath*}
T^{\mathrm{opt}}_{1}(a+1)-T^{\mathrm{opt}}_{1}(a)\geq \frac{2^{\lfloor \log a \rfloor+1}}{a} -\frac{2^{\lfloor \log a+1 \rfloor+1}}{a+1}.	
\end{dmath*} Using~\eqref{eq:E8}, we have 
\begin{dmath*}
T^{\mathrm{opt}}_{1}(a+1)-T^{\mathrm{opt}}_{1}(a)=\frac{2^{\lfloor \log a \rfloor+1}}{a(a+1)} > \frac{2^{\lfloor \log a \rfloor+1}}{a} -\frac{2^{\lfloor \log a+1 \rfloor+1}}{a+1},	
\end{dmath*} and subsequently,~\eqref{eq:E12} holds for $i=1$. Next, we assume that for $i=b$, we have 
\begin{dmath}\label{eq:E12-1}
T^{\mathrm{opt}}_{1}(a+b)-T^{\mathrm{opt}}_{1}(a)\geq \frac{2^{\lfloor \log a \rfloor+1}}{a} -\frac{2^{\lfloor \log a+b \rfloor+1}}{a+b}.	
\end{dmath}
It is enough to show that for $i=b+1$, we have 
\begin{dmath*}
T^{\mathrm{opt}}_{1}(a+b+1)-T^{\mathrm{opt}}_{1}(a)\geq \frac{2^{\lfloor \log a \rfloor+1}}{a} -\frac{2^{\lfloor \log a+b+1 \rfloor+1}}{a+b+1}.	
\end{dmath*}
We use proof by contradiction. Assume that
\begin{dmath}\label{eq:E12-2}
T^{\mathrm{opt}}_{1}(a+b+1)-T^{\mathrm{opt}}_{1}(a) < \frac{2^{\lfloor \log a \rfloor+1}}{a} -\frac{2^{\lfloor \log a+b+1 \rfloor+1}}{a+b+1}.	
\end{dmath} Combining~\eqref{eq:E12-1} and~\eqref{eq:E12-2}, we have \begin{equation}\label{eq:E12-3}
T^{\mathrm{opt}}_{1}(a+b+1)-T^{\mathrm{opt}}_{1}(a+b) < \frac{2^{\lfloor \log a+b \rfloor+1}}{a+b} -\frac{2^{\lfloor \log a+b+1 \rfloor+1}}{a+b+1}. 	
\end{equation} Using~\eqref{eq:E8}, we have 
\begin{dmath}\label{eq:E12-4}
T^{\mathrm{opt}}_{1}(a+b+1)-T^{\mathrm{opt}}_{1}(a+b)=\frac{2^{\lfloor \log (a+b)(a+b+1) \rfloor+1}}{a+b} \geq \frac{2^{\lfloor \log a+b \rfloor+1}}{a+b} -\frac{2^{\lfloor \log a+b+1 \rfloor+1}}{a+b+1}.	
\end{dmath} Putting~\eqref{eq:E12-3} and~\eqref{eq:E12-4} together, we arrive at a contradiction, and consequently,~\eqref{eq:E12} holds. 

Now, we can disprove~\eqref{eq:E11}. Taking $a = t$ and $i = 2k-2t-1$ in~\eqref{eq:E12}, we have 
\begin{dmath*}\label{eq:E12-5}
T^{\mathrm{opt}}_{1}(2k-t-1)-T^{\mathrm{opt}}_{1}(t)\geq \frac{2^{\lfloor \log t \rfloor+1}}{t} -\frac{2^{\lfloor \log 2k-t-1 \rfloor+1}}{2k-t-1}, 		
\end{dmath*} which readily contradicts~\eqref{eq:E11}. This completes the disproof of~\eqref{eq:E7}, and consequently, the proof for the case (i). 

\subsubsection*{Proof for Case (ii)} Noting that $l=2k+1$ and $\lfloor\frac{l}{2}\rfloor=k$,~\eqref{eq:E1} can be written as
\begin{dmath*}
\frac{\binom{2k}{k}}{\binom{2k+1}{k}}
T^{\mathrm{opt}}_{1}(k+1)+ \frac{\binom{2k}{k-1}}{\binom{2k+1}{k}}
T^{\mathrm{opt}}_{1}(k) \leq \frac{\binom{2k}{m}}{\binom{2k+1}{m}}
T^{\mathrm{opt}}_{1}(2k-m+1) + \frac{\binom{2k}{m-1}}{\binom{2k+1}{m}},
T^{\mathrm{opt}}_{1}(m)	
\end{dmath*} or equivalently,
\begin{dmath}\label{eq:E13}
\left(\frac{k+1}{2k+1}\right)T^{\mathrm{opt}}_{1}(k+1)+\left(\frac{k}{2k+1}\right)T^{\mathrm{opt}}_{1}(k) \leq \left(\frac{2k-m+1}{2k+1}\right)
T^{\mathrm{opt}}_{1}(2k-m+1) + \left(\frac{m}{2k+1}\right)
T^{\mathrm{opt}}_{1}(m) 	
\end{dmath} for all $1\leq m\leq k$. We use an inductive argument to prove that~\eqref{eq:E13} holds. It is easy to see that~\eqref{eq:E13} holds for $m=k$, i.e., 
\begin{dmath*}
\left(\frac{k+1}{2k+1}\right)T^{\mathrm{opt}}_{1}(k+1)+\left(\frac{k}{2k+1}\right)T^{\mathrm{opt}}_{1}(k) = \left(\frac{2k-k+1}{2k+1}\right)
T^{\mathrm{opt}}_{1}(2k-k+1) + \left(\frac{k}{2k+1}\right)
T^{\mathrm{opt}}_{1}(k).	
\end{dmath*} We assume that~\eqref{eq:E13} holds for $m>t$. We need to show that for $m=t$, the following holds:
\begin{dmath}\label{eq:E14}
\left(\frac{k+1}{2k+1}\right)T^{\mathrm{opt}}_{1}(k+1)+\left(\frac{k}{2k+1}\right)T^{\mathrm{opt}}_{1}(k) \leq \left(\frac{2k-t+1}{2k+1}\right)
T^{\mathrm{opt}}_{1}(2k-t+1) + \left(\frac{t}{2k+1}\right)
T^{\mathrm{opt}}_{1}(t)	
\end{dmath}. The proof is by contradiction. Suppose that~\eqref{eq:E14} does not hold, i.e.,
\begin{dmath}\label{eq:E15}
\left(\frac{k+1}{2k+1}\right)T^{\mathrm{opt}}_{1}(k+1)+\left(\frac{k}{2k+1}\right)T^{\mathrm{opt}}_{1}(k) > \left(\frac{2k-t+1}{2k+1}\right)
T^{\mathrm{opt}}_{1}(2k-t+1) + \left(\frac{t}{2k+1}\right)
T^{\mathrm{opt}}_{1}(t)		
\end{dmath}.
We need to disprove~\eqref{eq:E15}. Since~\eqref{eq:E13} holds for $m>t$ (by assumption), for $m=t+1$ we have
\begin{dmath}\label{eq:E16}
\left(\frac{k+1}{2k+1}\right)T^{\mathrm{opt}}_{1}(k+1)+\left(\frac{k}{2k+1}\right)T^{\mathrm{opt}}_{1}(k) \leq \left(\frac{2k-t}{2k+1}\right)
T^{\mathrm{opt}}_{1}(2k-t) + \left(\frac{t+1}{2k+1}\right)
T^{\mathrm{opt}}_{1}(t+1)		
\end{dmath}. Combining~\eqref{eq:E15} and~\eqref{eq:E16} yields
\begin{dmath*}
\left(\frac{2k-t+1}{2k+1}\right)
T^{\mathrm{opt}}_{1}(2k-t+1) + \left(\frac{t}{2k+1}\right)
T^{\mathrm{opt}}_{1}(t)		 < \left(\frac{2k-t}{2k+1}\right)
T^{\mathrm{opt}}_{1}(2k-t) + \left(\frac{t+1}{2k+1}\right)
T^{\mathrm{opt}}_{1}(t+1)	
\end{dmath*} which equivalently can be written as
\begin{dmath}\label{eq:E17}
\left(2k-t+1\right)
T^{\mathrm{opt}}_{1}(2k-t+1) -\left(2k-t\right)
T^{\mathrm{opt}}_{1}(2k-t) < \left(t+1\right)
T^{\mathrm{opt}}_{1}(t+1)-\left(t\right)
T^{\mathrm{opt}}_{1}(t)	
\end{dmath}.
Using~\eqref{eq:E8} and substituting $T^{\mathrm{opt}}_{1}(2k-t+1)$ by $ T^{\mathrm{opt}}_{1}(2k-t)+\frac{2^{\lfloor \log (2k-t) \rfloor+1}}{(2k-t+1)(2k-t)}$ and $T^{\mathrm{opt}}_{1}(t+1)$ by $ T^{\mathrm{opt}}_{1}(t)+\frac{2^{\lfloor \log t \rfloor+1}}{t(t+1)}$ in~\eqref{eq:E17}, we have
\begin{dmath*}
T^{\mathrm{opt}}_{1}(2k-t)+\frac{2^{\lfloor \log (2k-t) \rfloor+1}}{2k-t} <  T^{\mathrm{opt}}_{1}(t)+\frac{2^{\lfloor \log (t) \rfloor+1}}{t}	
\end{dmath*} which equivalently can be written as
\begin{dmath}\label{eq:E18}
T^{\mathrm{opt}}_{1}(2k-t)-T^{\mathrm{opt}}_{1}(t)< \frac{2^{\lfloor \log t \rfloor+1}}{t} -\frac{2^{\lfloor \log (2k-t) \rfloor+1}}{2k-t}.
\end{dmath} Taking $a = t$ and $i = 2k-2t$ in~\eqref{eq:E12}, we have 
\begin{dmath*}\label{eq:E18-1}
T^{\mathrm{opt}}_{1}(2k-t)-T^{\mathrm{opt}}_{1}(t)\geq \frac{2^{\lfloor \log t \rfloor+1}}{t} -\frac{2^{\lfloor \log 2k-t \rfloor+1}}{2k-t}, 		
\end{dmath*} which readily contradicts~\eqref{eq:E18}. This completes the disproof of~\eqref{eq:E15}, and consequently, the proof for the case (ii). 
\end{proof}

\begin{proof}[Proof of Lemma~\ref{average-lemma-2}]
By the results of Lemmas~\ref{lem:Tmw} and~\ref{lem:MidPoint}, the following recursive formulas can be shown: 
\begin{dmath*}
T^{\mathrm{opt}}_{1}(s) =  \frac{\lceil \frac{s}{2}\rceil}{s}T^{\mathrm{opt}}_{1}\left(\left\lceil\frac{s}{2}\right\rceil\right)+\frac{\lfloor\frac{s}{2}\rfloor}{s} T^{\mathrm{opt}}_{1}\left(\left\lfloor\frac{s}{2}\right\rfloor\right)+1,
\end{dmath*} for all $2\leq s\leq n$, and 
\begin{dmath*}
T^{\mathrm{opt}}_{2,1}(s) =  \frac{\lceil \frac{s}{2}\rceil}{s}T^{\mathrm{opt}}_{2,1}\left(\left\lceil\frac{s}{2}\right\rceil\right)+\frac{\lfloor\frac{s}{2}\rfloor}{s} T^{\mathrm{opt}}_{2,1}\left(\left\lfloor\frac{s}{2}\right\rfloor\right)+1,
\end{dmath*} and
\begin{dmath*}
T^{\mathrm{opt}}_{2,2}(s) =  \frac{\lceil \frac{s}{2}\rceil (\lceil \frac{s}{2}\rceil-1)}{s(s-1)}T^{\mathrm{opt}}_{2,2}\left(\left\lceil\frac{s}{2}\right\rceil\right)+\frac{\lfloor\frac{s}{2}\rfloor(\lfloor\frac{s}{2}\rfloor-1)}{s(s-1)} T^{\mathrm{opt}}_{2,2}\left(\left\lfloor\frac{s}{2}\right\rfloor\right)+\frac{2\lfloor\frac{s}{2}\rfloor \lceil\frac{s}{2}\rceil}{s(s-1)}\left(T^{\mathrm{opt}}_1\left(\left\lceil\frac{s}{2}\right\rceil\right)+T^{\mathrm{opt}}_1\left(\left\lfloor\frac{s}{2}\right\rfloor\right)\right)+1,
\end{dmath*} for all $3\leq s\leq n$. Solving these recursions, it follows that $T^{\mathrm{opt}}_1(2^i) = i$ for all $1\leq i\leq l$, and $T^{\mathrm{opt}}_{2,1}(2^i) = i$ and $T^{\mathrm{opt}}_{2,2}(2^i) = ({1+(i-1)(2^{i+1}-1)})/({2^i-1})$ for all ${2\leq i\leq l}$. Noting that $T^{\mathrm{opt}}_1(1)=0$, it is easy to verify that $T^{\mathrm{opt}}_1(2^i)=i$ for all $0\leq i\leq l$. Similarly, it follows that $T^{\mathrm{opt}}_2(s) = \frac{2}{s+1} T^{\mathrm{opt}}_{2,1}(s)+\frac{s-1}{s+1}T^{\mathrm{opt}}_{2,2}(s)$ for all $3\leq s\leq n$, and subsequently, $T^{\mathrm{opt}}_2(2^i) = ({(i-1)2^{i+1}+i+2})/({2^i+1})$ for all $2\leq i\leq l$. Noting that $T^{\mathrm{opt}}_2(1)=0$ and $T^{\mathrm{opt}}_2(2)=1$, it is easy to see that $T^{\mathrm{opt}}_2(2^i) = ({(i-1)2^{i+1}+i+2})/({2^i+1})$ for all $0\leq i\leq l$. This completes the proof.
\end{proof}

\end{document}